%% file: main.tex
\def\year{2022}\relax
\documentclass{article} %

\usepackage{aaai22}  %
\usepackage{times}  %
\usepackage{helvet}  %
\usepackage{courier}  %
\usepackage[hyphens]{url}  %
\usepackage{graphicx} %
\usepackage{natbib}  %
\usepackage{caption} %
\DeclareCaptionStyle{ruled}{labelfont=normalfont,labelsep=colon,strut=off} %
\usepackage{algorithm}

\usepackage[hidelinks]{hyperref}

\usepackage{newfloat}
\usepackage{listings}
\floatstyle{ruled}
\newfloat{listing}{tb}{lst}{}
\floatname{listing}{Listing}

\title{Team Correlated Equilibria in Zero-Sum Extensive-Form Games \\ via Tree Decompositions}
\affiliations{
    \textsuperscript{\rm 1}Computer Science Department, Carnegie Mellon University\\
    \textsuperscript{\rm 2}Strategic Machine, Inc.\\
    \textsuperscript{\rm 3}Strategy Robot, Inc.\\
    \textsuperscript{\rm 4}Optimized Markets, Inc. \\
    \{bhzhang, sandholm\}@cs.cmu.edu \\
}
\author{
    Brian Hu Zhang,\textsuperscript{\rm 1}
    Tuomas Sandholm\textsuperscript{\rm 1,2,3,4}
}

\usepackage{amsthm}

\usepackage{amsmath}
\usepackage{graphicx}
\usepackage{mathtools}
\usepackage{multicol}
\usepackage{listings}
\usepackage[shortlabels,inline]{enumitem}
\usepackage[italicdiff]{physics}
\usepackage{cancel}
\usepackage{color}
\usepackage{verbatim}
\usepackage{comment}
\usepackage{array}
\usepackage{soul}
\usepackage{tikz-cd}
\usepackage{accents}
\usepackage{bbm}
\usepackage{xspace}
\usepackage{siunitx}
\usepackage{linegoal}
\usepackage{diagbox}

\usepackage{hhline}
\usepackage{comment}
\usepackage{array}
\usepackage{booktabs}

\usepackage{aliascnt}
\usepackage{algorithm}
\usepackage[noend]{algpseudocode}
\usepackage{amssymb}
\usepackage{minibox}

\usepackage{tikz}
\usepackage{tikz-qtree}
\usetikzlibrary{decorations.pathmorphing}

\usepackage{cleveref}
\usepackage{autonum} 
\makeatletter
\autonum@generatePatchedReferenceCSL{Cref}
\autonum@generatePatchedReferenceCSL{cref}
\makeatother

\definecolor{green}{rgb}{0,0.5,0}

\newcommand{\delimit}[3]{\newcommand{#1}[1]{\left#2##1\right#3}}
\delimit \ceil \lceil \rceil
\delimit \floor \lfloor \rfloor

\DeclareMathOperator*{\E}{\mathbb E}

\let\op\operatorname

\let\oldboxed\boxed
\renewcommand{\boxed}[1]{\oldboxed{#1}\,}

\newcommand{\zo}{\{0, 1\}}
\newcommand{\R}{\mathbb R}

\newcommand{\X}{{\mc X}}
\newcommand{\Y}{{\mc Y}}

\newcommand{\pmax}{\ensuremath{\oplus}\xspace}
\newcommand{\pmin}{\ensuremath{\ominus}\xspace}
\newcommand{\nature}{\textsc{Nature}\xspace}

\let\vec\boldsymbol

\newcommand{\team}{T}
\newcommand\cabove{\blacktriangle}
\newcommand\cbelow{\blacktriangledown}

\renewcommand\grad\nabla

\let\mc\mathcal

\numberwithin{equation}{section}

\newtheorem{theorem}[equation]{Theorem}

\newtheorem{proposition}[equation]{Proposition}
\newtheorem{corollary}[equation]{Corollary}

\theoremstyle{definition}
\newtheorem{definition}[equation]{Definition}

\setlist[enumerate,1]{label={(\arabic*)}}

\newcommand{\commentsymbol}{\it\color{gray}$\triangleright$~}
\algrenewcommand\algorithmiccomment[1]{\hfill{\commentsymbol#1}}

\makeatletter
\let\c@algorithm\relax
\makeatother
\newaliascnt{algorithm}{equation}

\begin{document}

\maketitle

\begin{abstract}
Despite the many recent practical and theoretical breakthroughs in computational game theory, equilibrium finding in extensive-form team games remains a significant challenge. While {\sf NP}-hard in the worst case, there are provably efficient algorithms for certain families of team game. In particular, if the game has {\em common external information}, also known as {\em A-loss recall}---informally, actions played by non-team members (i.e., the opposing team or nature) are either unknown to the entire team, or common knowledge within the team---then polynomial-time algorithms exist~\cite{Kaneko95:Behavior}. In this paper, we devise a completely new algorithm for solving team games. It uses a tree decomposition of the constraint system representing each team's strategy to reduce the number and degree of constraints required for correctness (tightness of the mathematical program). Our approach has the bags of the tree decomposition correspond to team-public  states---that is, minimal sets of nodes  (that is, states of the team) such that, upon reaching the set, it is common knowledge  among the players on the team that the set has been reached. Our algorithm reduces the problem of solving team games to a linear program with at most $O(NW^{w+1})$ nonzero entries in the constraint matrix, where $N$ is the size of the game tree, $w$ is a parameter that depends on the amount of {\em uncommon} external information, and $W$ is the treewidth of the tree decomposition. In {\em public-action games}, our program size is bounded by the tighter $2^{O(nt)}N$ for teams of $n$ players with $t$ types each. Our algorithm is based on a new way to write a custom, concise tree decomposition, and its fast run time does not assume that the decomposition has small treewidth. Since our algorithm describes the polytope of correlated strategies directly, we get equilibrium finding in correlated strategies for free---instead of, say, having to run a double oracle algorithm. We show via experiments on a standard suite of games that our algorithm achieves state-of-the-art performance on all benchmark game classes except one. We also present, to our knowledge, the first experiments for this setting where both teams have more than one member.
\end{abstract}

\section{Introduction}
Computational game solving in imperfect-information games has led to many recent superhuman breakthroughs in AI~(e.g.,~\citealp{Bowling15:Heads,Brown18:Superhuman,Brown19:Superhuman}). Most of the literature on this topic focuses on {\em two-player zero-sum games with perfect recall}---that is, two players who {\em never forget any information during the game} face off in an adversarial manner. While this model is broad enough to encompass games such as poker, it breaks in the setting of {\em team games}, in which two {\em teams} of players face off against each other. While members of the team have perfect recall, the team as a whole may not, because different members of the team may know different pieces of information. Situations that fall into this category include recreational games such as contract bridge, Hanabi (in which there is no adversary), collusion in poker, and all sorts of real-world scenarios in which communication is limited. In this paper, we focus on such games.

We will assume that team members can coordinate before the game begins, including generating randomness that is shared within the team but hidden from the opposing team; in other words, they can {\em correlate} their strategies. Once the game begins, they can only exchange information by playing actions within the game. An alternative way of interpreting the setting is to consider zero-sum games with {\em imperfect recall} (e.g., \citealp{Kaneko95:Behavior}), that is, games in which players may forget information that they once knew. In this interpretation, each team is represented by a single player whose ``memory'' is continuously refreshed to that of the acting team member. The two interpretations are equivalent.

In general, computing equilibria in such games is {\sf NP}-hard~\cite{Chu01:NP}. However, some subfamilies of games are efficiently solvable. For example, if both teams have {\em common external information}, also known as {\em A-loss recall}~\cite{Kaneko95:Behavior} (in which all uncommon knowledge of the team can be attributed to team members not knowing about actions taken by other team members), or both teams have at most two players and the interaction between team members is {\em triangle-free}~\cite{Farina20:Polynomial} (which roughly means that the team's strategy tree can be re-ordered into a strategically equivalent tree obeying A-loss recall), then each player's strategy space can be described as the projection of a polytope with polynomially many constraints in the size of the game, and hence the game can be solved in polynomial time.

Practical algorithms for solving medium-sized instances of these games~\cite{Celli18:Computational,Farina18:ExAnte,Zhang20:Computing,Farina21:Connecting} primarily focused on the case in which there is a team of players playing against a single adversary. These algorithms are mostly based on {\em column generation}, or {\em single oracle}, and require a best-response oracle that is implemented in practice by an integer program. While they perform reasonably in practice, they lack theoretical guarantees on runtime. Although these techniques can be generalized naturally to the case of two teams using {\em double oracle}~\cite{McMahan03:Planning} in place of their column generation methods, we do not know of a paper that explores this more general case.

In this paper, we demonstrate a completely new approach to solving team games. From a game, we first construct a custom, concise {\em tree decomposition} (\citealp{Robertson86:Graph}; for a textbook description, see \citealp{Wainwright08:Graphical}) for the constraint system that defines the strategy polytope of each player. Then, we bound the number of feasible solutions generated in each tree decomposition node, from which we derive a bound on the size of the representation of the whole strategy space. Our bound is linear in the size of the game tree, and exponential in a natural parameter $w$ that measures the amount of uncommon external information. We also show a tighter bound in games with public actions (such as poker). 
Since our algorithm describes the polytope of correlated strategies directly, we get equilibrium finding in correlated strategies for free---instead of, say, having to run a double oracle algorithm.  %
We show via experiments on a standard suite of games that our algorithm is state of the art in most game instances, with failure modes predicted by the theoretical bound. 
We also present, to our knowledge, the first experiments for the setting where both teams have more than one member. 

Some papers~(e.g.,~\citealp{Daskalakis06:Computing}) have explored the use of tree decompositions to solve {\em graphical games}, which are general-sum, normal-form games in which the interactions between {\em players} are described by a graph. Our setting, and thus the techniques required, are completely different from this line of work.

\section{Preliminaries}
In this section, we will introduce background information about extensive-form games and tree decompositions.
\subsection{Extensive-Form Games}
\begin{definition}
A {\em zero-sum extensive-form team game} (EFG) $\Gamma$, hereafter simply {\em game}, between two teams $\pmax$ and $\pmin$ consists of the following: 
\begin{enumerate}
    \item A finite set $H$, of size $\abs{H} := N$, of {\em histories} of vertices of a tree rooted at some {\em initial state} $\textsc{Root} \in H$. The set of leaves, or {\em terminal states}, in $H$ will be denoted $Z$. The edges connecting any node $h \in H$ to its children are labeled with {\em actions}. The child node created by following action $a$ at node $h$ will be denoted $ha$. 
    \item A map $P : (H \setminus Z) \to \{\pmax, \pmin, \nature\}$, where $P(h)$ is the team who acts at node $h$. A node at which team $\team$ acts is called an {\em $\team$-node}.
    \item A {\em utility function} $u : Z \to \R$.
    \item For each team $\team \in \{\pmax, \pmin\}$, a partition $\mc I_\team$ of  the set of $\team$-nodes into \textit{information sets}, or {\em infosets}. In each infoset $I \in \mc I_\team$, every pair of nodes $h, h' \in I$ must have the same set of actions. 
    \item For each nature node $h$, a distribution $p(\cdot |h)$ over the actions available to nature at node $h$.
\end{enumerate} 
\end{definition}
It will sometimes be convenient to discuss the individual {\em players} on a team. In this context, for each team $\team \in \{\pmax, \pmin\}$, we will assume that $\team$ itself is a set of distinct players, and there is a map $P_\team : \mc I_\team \to \team$ denoting which member of the team plays at each infoset $I \in \mc I_\team$. 

We will use the following notational conventions: $A(h)$ or $A(I)$ denotes the set of available actions at a node $h$ or infoset $I$. $\preceq$ denotes the partial order created by the tree: if $h, h'$ are nodes, infosets or sequences, $h \preceq h'$ means $h$ is an ancestor of $h'$ or $h' = h$. If $S$ is a set of nodes, $h \succeq S$ (resp. $h \preceq S$) means $h \succeq h'$ (resp. $h \preceq h'$) for some $h' \in S$. If $I$ is an infoset and $a$ is an action at $I$, then $Ia = \{ ha : h \in I \}$. $\land$ denotes the lowest common ancestor relation: $h \land h'$ is the deepest node $h^*$ of the tree for which $h^* \preceq h, h'$. At a given node $h$, the {\em sequence} $s_i(h)$ for a team or player $i$ is the list of infosets reached and actions played by $i$ up to node $h$, including the infoset at $h$ itself if $i$ plays at $h$. %

We will assume that each individual player on a team has perfect recall, that is, for each player $i \in \team$, each infoset $I$ with $P_\team(I) = i$, and each pair of nodes $h, h' \in I$, we must have $s_i(h) = s_i(h')$. Of course, the team as a whole may not have perfect recall.
We will also assume that the game is {\em timeable}, {\em i.e.,}, every node in a given infoset is at the same depth of the game tree. While this assumption is not without loss of generality, most practical games satisfy it.

A {\em pure strategy} $\sigma$ for a team $\team \in \{\pmax, \pmin\}$ is a selection of one action from the action space $A(I)$ at every infoset $I \in \mc I_\team$.  The {\em realization plan}~\cite{Farina18:ExAnte}, or simply {\em plan}, $\vec x_\sigma$ corresponding to $\sigma$ is the vector $\vec x^\sigma \in \zo^H$ where $x^\sigma_h = 1$ if $\sigma(I) = a$ for every $Ia \preceq h$. A {\em correlated strategy} $\tilde \sigma$ is a distribution over pure strategies. The plan $\vec x$ corresponding to $\tilde \sigma$ is the vector $\vec x \in  [0, 1]^{H}$ where $x_h = \E_{\sigma \sim \tilde \sigma} x^\sigma_h$. The spaces of plans for teams $\pmax$ and $\pmin$ will be denoted $\X$ and $\Y$ respectively.  Both $\X$ and $\Y$ are compact, convex polytopes.

A {\em strategy profile} is a pair $(\vec x, \vec y) \in \X \times \Y$. The {\em value} of the strategy profile $(\vec x, \vec y)$ for \pmax is $u(\vec x, \vec y) := \sum_{z \in Z} u_z p_z x_z y_z$, where $p_z$ is the probability that nature plays all the actions needed to reach $z$: $p_z := \prod_{ha \preceq z : P(h) = \nature} p(a|h)$. Since the game is zero-sum, the payoff for \pmin is $-u(\vec x, \vec y)$. The {\em best response value} for \pmin (resp. \pmax) to a strategy $\vec x \in \X$ (resp. $\vec y \in \Y$) is $u^*(\vec x) := \min_{\vec y \in \Y} u(\vec x, \vec y)$ (resp. $u^*(\vec y) := \max_{\vec x \in \X} u(\vec x, \vec y)$).  A strategy profile $(\vec x, \vec y)$ is a {\em team correlated equilibrium}, or simply {\em equilibrium}, if $u^*(\vec x) = u(\vec x,\vec y) = u^*(\vec y)$. Every equilibrium of a given game has the same value, which we call the {\em value of the game}.
Equilibria in zero-sum team games can be computed by solving the bilinear saddle-point problem
\begin{align}\label{eq:bspp}
    \max_{\vec x \in \X} \min_{\vec y \in \Y} u(\vec x, \vec y),
\end{align}
where $u(\vec x, \vec y)$ can be expressed as a bilinear form $\ev{\vec x, \vec A\vec y}$ in which $\vec A$ has $O(N)$ nonzero entries.

One may wonder why we insist on allowing players to correlate. Indeed, alternatively, we could have defined {\em uncorrelated} strategies and equilibria, in which the distribution over pure strategies of each team is forced to be a product distribution over the strategy spaces of each player. However, in this case, the strategy space for both players becomes nonconvex, and therefore we may not even have equilibria at all! Indeed, if $\tilde \X$ and $\tilde \Y$ are the spaces of plans for uncorrelated strategies for each team, Theorem~7 in \citet{Basilico17:Team} implies that it is possible for
\begin{align}
    \max_{\vec x \in \tilde \X} \min_{\vec y \in \tilde \Y} u(\vec x, \vec y) \ne \min_{\vec x \in \tilde \Y} \max_{\vec x \in \tilde \X} u(\vec x,\vec y),
\end{align}
which makes it difficult to {\em define} the problem, much less solve it. Some authors (e.g.,~\citealp{Basilico17:Team}) have defined ``team maxmin equilibria'' in these games for the case where there is one opponent, i.e., $\abs{\pmin}=1$, by assuming that the team \pmax plays first. In the present paper, we study zero-sum games between two {\em teams}---as such, due to symmetry between the teams, we have no reason to favor one team over the other, and thus cannot make such an assumption.

If $\X$ and $\Y$ can be described succinctly by linear constraints, the problem~\eqref{eq:bspp} can be solved by taking the dual of the inner minimization and solving the resulting linear program~\cite{Luce57:Games}. Unfortunately, such a description cannot exist in the general case unless {\sf P} = {\sf NP}:
\begin{theorem}[\citealp{Chu01:NP}]\label{th:hardness}
Given a zero-sum team game and a threshold value $u^*$, determining whether the game's value is at least $u^*$ is {\sf NP}-hard, even when \pmax has two players and there is no adversary.
\end{theorem}
For completeness, we include proofs of results in this section in the appendix.
Despite the hardness result, it is sometimes possible to express $\X$ and $\Y$ using polynomially many constraints, and hence solve~\eqref{eq:bspp} efficiently. One of these cases is the following:
\begin{definition}
A team $\team$ in a game has {\em common external information}, also known as {\em A-loss recall}~\cite{Kaneko95:Behavior}, if, for any two nodes $h, h'$ in the same infoset $I$ of team $T$, either $h_\team = h'_\team$, or there exists an infoset $I'$ and two actions $a \ne a'$ at $I'$ such that $h \succeq I'a$ and $h' \succeq I'a'$. 
\end{definition}
Informally, if a team has common external information, each member of that team has the same knowledge about actions taken by non-team members (i.e., the opponent and nature). In this case, the {\em perfect-recall refinement} of that team's strategy space, which is created by splitting infosets to achieve perfect recall, is strategically equivalent with respect to correlation plans. Thus, equilibria can be found efficiently when both teams have it.

In this paper, we expand and refine these complexity results. In particular, we show that the polytope of plans for a team can be expressed as a projection of a polytope with $O(NW^{w+1})$ constraints and variables, where $w$ is a parameter describing the amount of information {\em generated by non-team members} that is known by {\em at least one member, but not all members}, of the team, and $W$ is the treewidth of a particular tree decomposition. In particular, $w=1$ for games with common external information, so our result is equivalent to the known result in that case.

\subsection{Tree Decompositions, Integer Hulls, and Dependency Graphs}\label{se:dependency-graphs}
We now review tree decompositions and their use in parameterizing integer hulls. We refer the reader to~\citet{Wainwright08:Graphical} for further reading on these topics.
\begin{definition}
Let $G$ be a graph. A {\em tree decomposition}, also known as a {\em junction tree} or {\em clique tree}, is a tree $\mc J$, whose vertices are subsets of vertices of $G$, called {\em bags}, such that
\begin{enumerate}
\item for every edge $(u, v)$ in $G$, there is a bag in $\mc J$ containing both $u$ and $v$, and
\item for every vertex\footnote{We will use $u \in G$ to mean $u$ is a vertex in $G$.} $u \in G$, the subset of bags $\{ U \in \mc J : u \in U \}$ is nonempty and connected in $\mc J$.
\end{enumerate} 
\end{definition}
The {\em width} of $\mc J$ is the size of the largest bag, minus one.\footnote{The $-1$ is so that trees have treewidth $1$.}
Now, consider a convex set of the form\footnote{$\op{conv}$ denotes the convex hull operator. $[m] = \{1, \dots, m\}$.}
\begin{gather}
\begin{aligned}
    {\mc X} &= \op{conv}(X), \qq{where}
    \\ X &= \qty{\vec x \in \zo^n : g_k(\vec x) = 0~~\forall k \in [m]} 
\end{aligned}\label{eq:int-hull}
\end{gather}
and the constraints $g_k : \zo^n \to \R$ are arbitrary polynomials. This formulation is very expressive; for example, it is possible to express realization plan polytopes in this form where $X$ is the set of pure plans, as we will show later. For notation, in the rest of the paper, if $\vec x \in \R^n$ and $U \subseteq [n]$, we will use $\vec x_U$ to denote the subvector of $\vec x$ on the indices in $U$.

\begin{definition}
The {\em dependency graph} of $\mc X$ is the graph whose vertices are the indices $i \in [n]$, and for which an edge connects $i, j \in [n]$ if there is a constraint $g_k$ in which both variables $x_i$ and $x_j$ appear.
\end{definition}

\begin{definition}
Let $U \subseteq [n]$, and $\tilde{\vec x} \in \{0, 1\}^U$. We say that $\tilde{\vec x}$ is {\em locally feasible} on $U$ if $\tilde{\vec x} = \vec x_U$ for some $\vec x \in X$. 
\end{definition}
We will use $X_U$ to denote the set of locally feasible vectors on $U$. Of course, $X_{[n]} = X$. The following result follows from the junction tree theorem (e.g.,~\citealt{Wainwright08:Graphical}):

\begin{proposition}\label{le:jtt-csp}%
Let $\mc J$ be a tree decomposition of the dependency graph of $\mc X$. Then $\vec x \in \mc X$ if and only if there exist vectors\footnote{$\Delta^S$ is the set of distributions on set $S$} $\vec\lambda^U \in \Delta^{X_U}$ for each $U \in \mc J$ satisfying
{\rm
\newcommand{\linetext}[1]{\parbox[t]{10em}{#1}}
\begin{align}
    \vec x_U &= \sum_{\tilde{\vec x} \in X_U} \lambda^U_{\tilde{\vec x}} \tilde{\vec x} &&\linetext{$\forall$ bags $U \in \mc J$} \raisetag{15pt}\stepcounter{equation}\tag{\theequation}\label{eq:generic-tree-decomposition-constraints}\\
    \sum_{\substack{\tilde{\vec x} \in X_U : \\ \tilde{\vec x}_{U \cap V} = \tilde{\vec x}^*}} \lambda^U_{\tilde{\vec x}} &= \sum_{\substack{\tilde{\vec x} \in X_V : \\ \tilde{\vec x}_{U \cap V} = \tilde{\vec x}^*}} \lambda^V_{\tilde{\vec x}} &&\linetext{$\forall$ edges $(U, V)$ in $\mc J$, \\\phantom\quad locally feasible  \\\phantom\quad\quad $\tilde{\vec x}^* \in X_{U \cap V}$.}
\end{align}
}%
The constraint system has size\footnote{Throughout this paper, {\em size} of a program refers to the number of nonzero entries in its constraint matrix.} $O((W+D) \sum_{U \in \mc J} \abs{X_U})$, where $W$ is the treewidth of $\mc J$ and $D$ is the maximum degree of any bag in $J$.
\end{proposition}
\Cref{le:jtt-csp} establishes a method of parameterizing convex combinations $\vec x$ over a set $X \subseteq \{0,1\}^n$. First, write down a tree decomposition of the constraint system defining $X$. Then, for each bag $U$ in the tree decomposition, parameterize a {\em local distribution}  over the set of locally feasible solutions $X_U$, and insist that, for adjacent bags $U, V$ in $\mc J$, the distributions $\vec\lambda^U$ and $\vec\lambda^V$ agree on the marginal on $U \cap V$. Thus, given an arbitrary set $\mc X$ of the given form, to construct a small constraint system that defines $\mc X$, it suffices to construct a tree decomposition of its constraint graph.

In most use cases of tree decompositions, the next step is to bound the treewidth, and then use the bound $\abs{X_U} \le 2^{W+1}$ to derive the size of the constraint system \eqref{eq:generic-tree-decomposition-constraints}. In our domain, it will turn out that $W$ can be too large to be useful; hence, we will instead {\em directly} bound $\abs{X_U}$.

\section{Tree Decompositions for Team Games}

We now move to presenting our results. 
We take the perspective of the maximizing team $\pmax$. All of the results generalize directly to the opposing team, $\pmin$, by swapping \pmax for \pmin and $X, \X$ for $Y, \Y$. 

Our approach has the bags of the tree decomposition  correspond  to team-public  states---that  is,  minimal sets of  nodes  (that is, states of the team) such  that,  upon reaching the set,  it  is  common  knowledge  among  the  players  on  the  team  that  the set has  been reached.  This is similar, but not identical, to the notion of public tree, which instead decomposes based on knowledge common to all players, including opponents. Next we will present our approach formally.

We will assume that no two nodes $h \ne h'$ at the same depth of the tree have ${\sf seq}_\pmax(h) = {\sf seq}_\pmax(h')$. That is, there is no information known by no members of the team. We will also assume that the information partition of \pmax has been {\em completely inflated}~\cite{Kaneko95:Behavior}---in particular, we will assume that, if \pmax has common external information, then \pmax has perfect recall.
These assumptions can be satisfied without loss of generality by merging nodes with the same sequence in the game tree, and performing inflation operations as necessary, before the tree decomposition.

Before defining the tree decomposition, we first need to define a notion of {\em public state}. Let $\sim$ be the following relation on pairs of nonterminal nodes: $h_1 \sim h_2$ if $h_1$ and $h_2$ are at the same depth {\em and} there is an infoset $I \in \mc I_\pmax$ with $h_1, h_2 \preceq I$. Let $\mc J^{\cabove}_\pmax$ be the set of equivalence classes of the transitive closure of $\sim$. The elements of $\mc J^{\cabove}_\pmax$ are the {\em public states} of the team:
each public state is a minimal set of nodes such that, upon reaching that set, it is common knowledge among the team that the team's true node lies within that set.
That is, if the true history of the team is in some $U^\cabove \in \mc J_\pmax^{\cabove}$, then this fact is common knowledge among the team. 
Our notion differs from the usual notion (e.g.,~\citealt{Kovavrik21:Rethinking}), in two ways. First, our public states here are {\em restricted to the team}: the fact that a team is in a public state may not be known or common knowledge among members of the opposing team. Second, our public states are constructed from the game tree rather than supplied as part of the game description---as such, they may capture common knowledge that is not captured by public observations.

The space of feasible plans can be written as $\mc X = \op{conv}(X)$, where $X$ is the set of pure plans $\vec x \in \{0, 1\}^{H}$, that is, plans satisfying:
{
\newcommand{\linetext}[1]{\parbox[t]{9em}{#1}}
\begin{gather}
\begin{aligned}
&x_h = \sum_{a \in A(h)} x_{ha}
&&\linetext{$\forall$ \pmax-nodes $h$}
\\ &x_h = x_{ha}
&&\linetext{$\forall$ non-\pmax-nodes $h$,
\\\phantom\quad actions $a \in A(h)$}
\\&x_{ha} x_{h'} = x_{h'a} x_h 
&&\linetext{$\forall$ infosets $I \in \mc I_\pmax$, 
\\\phantom\quad nodes $h, h' \in I$,
\\\phantom\quad actions $a \in A(I)$} 
\end{aligned}\label{eq:team-opt}
\end{gather}
}

We will never explicitly write a program using this constraint system; the only purpose of defining it is to be able to apply \Cref{le:jtt-csp} to the resulting tree decomposition.

We now construct a tree decomposition of the dependency graph of~\eqref{eq:team-opt}. 
For a public state $U^\cabove \in \mc J^{\cabove}_\pmax$, let $U^\cbelow$ be the set of all children of nodes in $U^\cabove$, that is, $U^\cbelow := \{ ha : h \in U, a \in A(h)\}$. Let $U = U^\cabove \cup U^\cbelow$.

\begin{figure}[tb]
    \centering
    \input{picture.tex}
    \caption{An example of a game with a single team, and its tree decomposition. Utilities are not shown as they are not relevant. Player nodes are shaded; the root (A) is a nature node. Information sets are joined by dotted lines. In each bag $U = U^\cabove \cup U^\cbelow$, the nodes in $U^\cabove$ are shown in the first row, and the nodes in $U^\cbelow$ are shown in the second row.}
    \label{fi:example}
\end{figure}
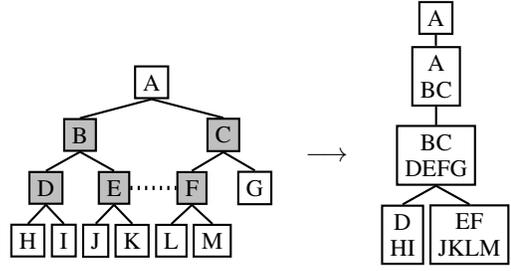

\begin{theorem}\label{th:tree-decomp}
The following is a tree decomposition $\mc J_\pmax$ of the constraint system~\eqref{eq:team-opt}.
\begin{enumerate}
\item The bags are the sets $U$ for $U^\cabove \in \mc J^{\cabove}_\pmax$, and the set $\{\textsc{Root}\}$ containing only the root node.
\item There is an edge between bags $U$ and $V$ if $U \cap V \ne \emptyset$
\end{enumerate}
\end{theorem}
\begin{proof}
We first check that $\mc J_\pmax$ is actually a tree. Since public states $U^\cabove \in \mc J_\pmax^{\cabove}$ are disjoint sets, edges in $\mc J_\pmax$ must span across different depths. Therefore, it suffices to show that for every $U \in \mc J_\pmax$, there is only (at most) one edge from $U$ to any bag at shallower depth. Let $u, v \in U^\cabove$, and $u'$ and $v'$ be the parents of $u$ and $v$ respectively. It suffices to show that $u'$ and $v'$ are in the same public state. Since $u$ and $v$ are in the same public state, there is a sequence of nodes $v_0, \dots, v_k \in U^\cabove$ such that $u = v_0 \sim v_1 \sim \dots \sim v_k = v$. Let $v_i'$ be the parent of $v_i$. Then by definition of $\sim$, we have $u' = v_0' \sim v_1' \sim \dots \sim v_k' = v'$.

Now by definition of public state,  any two nodes in the same infoset share the same public state, and for a nonterminal node $h \in U^\cabove$, $U$ by definition contains both $h$ and all its children. Therefore, every constraint is contained in some bag. Finally, every node in the game tree appears in at most two bags $U$ and $U'$, where $U'$ is the parent of $U$ in the tree decomposition, and these are connected by an edge. We have thus checked all the required properties of a tree decomposition, so we are done.
\end{proof}
Therefore, it remains only to construct the sets $X_U$ of locally feasible solutions on each $U$, and bound their sizes. The tree structure of $\mc J_\pmax$ is the {\em public tree} for the team.
and an example can be found in \Cref{fi:example}.
\begin{algorithm}[htb]
\caption{Constructing locally feasible sets }\label{al:local-feasible}
\begin{algorithmic}[1]
\For{each $U \in \mc J_\pmax$, in top-down order}
\If{$U=\{\textsc{Root}\}$} set $X_U \gets \{1\}$ and {\bf continue}\EndIf
\State let $U'$ be the parent of $U$ {\color{gray}\em (by construction, $U^\cabove \subset U'$)}

\State $X_{U^\cabove} \gets \{ \tilde{\vec x}_{U^\cabove} : \tilde{\vec x} \in X_{U'}\}$
\State $X_{U} \gets \emptyset$
\For{each locally feasible solution $\tilde{\vec x} \in X_{U^\cabove}$}
    \State $\mc I_{\tilde{\vec x}} \gets \{ I \in \mc I_\pmax : I \subseteq U^\cabove, I \cap \tilde{\vec x}^{-1}(1) \ne \emptyset\}$.
    \For{each partial strategy ${\vec a} \in \bigtimes_{I \in \mc I_{\tilde{\vec x}}} A(I)$} 
        \State $\tilde{\vec x}' \gets \vec 0 \in \zo^U$
        \For{each $h \in U^\cabove$ such that $\tilde x_h = 1$}
            \State $\tilde x'_h \gets 1$
            \If{$h \in I$ is a \pmax-node}
                $\tilde x'_{ha_I} \gets 1$
            \Else~{\bf for} each $a \in A(h)$ {\bf do} $\tilde x'_{ha} \gets 1$
            \EndIf
        \EndFor
        \State add $\tilde{\vec x}'$ to $X_U$.
    \EndFor
\EndFor
\EndFor
\end{algorithmic}
\end{algorithm}

\Cref{al:local-feasible} enumerates the locally feasible sets $X_U$ in each bag $U \in \mc J_\pmax$ iteratively starting from the root. It has runtime $O(W\sum_{U \in \mc J_\pmax} \abs{X_U})$, where $W$ is the treewidth\footnote{The degree of any bag $U$ in the decomposition is at most $\abs{U} \le W+1$, because each child of $U$ must contain a disjoint subset of $U^\cbelow$, so in \Cref{le:jtt-csp} we have $D = O(W)$ and we can therefore ignore $D$.} and a straightforward induction shows that it is correct. 
Therefore, chaining together \Cref{al:local-feasible} and \eqref{eq:generic-tree-decomposition-constraints} to obtain a description of both players' polytopes $\X$ and $\Y$, and solving the resulting bilinear saddle-point problem by dualizing the inner minimization and using linear programming, we obtain:
\begin{theorem}[Main Theorem]\label{th:main1}
Team correlated equilibria in extensive-form team games can be computed via a linear program of size 
\begin{align}
    O\qty(N + W\sum_{U \in \mc J_\pmax} \abs{X_U} + W\sum_{U \in \mc J_\pmin} \abs{Y_{U}}).
\end{align}
\end{theorem}

In \Cref{se:full-program}, we include the full constraint system for a player's strategy space in both the general case and, as an example, the game in \Cref{fi:example}.

For intuition, we briefly discuss the special case where the team has perfect recall. In this case, after applying the assumptions without loss of generality, the team tree is (up to some redundancies) exactly the sequence-form tree for the team. Every public state $U^\cabove \in \mc J^\cabove_\pmax$ has exactly one node, namely the information set node of the lone information set $I$ at $U^\cabove$, and the local feasible solutions $\tilde{\vec x} \in X_U$ correspond to sequences $Ia$ for actions $a$ at $I$. Thus, the LP given by \Cref{th:main1} is (again up to some redundancies) exactly the sequence-form LP~\cite{Koller94:Fast}.

\section{Bounding the Sizes of Locally Feasible Sets}
On its own, \Cref{th:main1} is not very useful: we still need to bound the sizes $\abs{X_U}$ and $\abs{Y_{U}}$. In the worst case, we will not be able to derive small bounds due to {\sf NP}-hardness. However, we will now show some cases in which this result matches or tightens known results. In this section, we again take the perspective of a single team \pmax. %

\subsection{Low Uncommon External Information}

Call a set $S \subseteq U^\cabove$ {\em reachable} if there is a pure plan $\vec x \in X$ such that $S$ is exactly the set of nodes in $U^\cabove$ reached by $\vec x$, i.e., $S = \{ h \in U^\cabove : x_h = 1\}$. Let $w(U)$ be the largest size of any reachable subset of $U^\cabove$. If $w(U)$ is small for every $U$, then we can bound the size of the linear program:
\begin{theorem}\label{co:low-external-info-loss}
Team correlated equilibria in extensive-form team games can be computed via a linear program of size 
\begin{align}
    O\qty(N + W \sum_{U \in \mc J_\pmax \cup \mc J_\pmin} \abs{U}^{w(U)}) \le O(NW^{w+1}),
\end{align}
where we call $w := \max_{U \in \mc J_\pmax \cup J_\pmin} w(U)$ the {\em reachable width} of $\mc J_\pmax$.
\end{theorem}

\begin{proof}
Let $\vec x \in X$, and let $U \in J_\pmax$. Let $$U^* = \{ha \in U^\cbelow : P(h) = \pmax\} \cup \{h \in U^\cabove : P(h) \ne \pmax\}.$$
Then the following are true:
\begin{enumerate}[(1)]
    \item $\vec x_{U^*}$ uniquely determines $\vec x_{U}$: for a \pmax-node $h \in U^\cabove$, we have $x_h = \sum_a x_{ha}$, and for a non-\pmax-node $h \in U^\cabove$, we have $x_h = x_{ha}$ for all $a$.  \label{it:unique}
    \item Let $h, h' \in U^*$, and suppose $x_h = x_{h'} = 1$. Then $h \land h'$ is not a $\pmax$-node: otherwise, a pure strategy playing to both $h$ and $h'$ would have to select two actions at $h \land h'$.
\end{enumerate}

Thus, we have, $\abs{X_U} \le \binom{\abs{U^*}}{\le w(U)} \le \abs{U}^{w(U)}$, where $\binom{n}{\le k} := \sum_{j=0}^k \binom{n}{k}$ is the number of ways to pick at most $k$ items from a set of size $n$. The theorem follows.
\end{proof}

This bound is applicable in any game, but, again due to {\sf NP}-hardness in the worst case, there will exist games in which $w = \Theta(N)$, in which case the bound will be exponentially bad and we would rather use the trivial bound $\abs{X_U} \le 2^W$.  

We now state some straightforward properties of reachable sets $S \subseteq U^\cabove$. 
\begin{enumerate}
    \item Every pair of nodes $h \ne h' \in S$ has a different sequence $s_\pmax(h)$. That is, information distinguishing nodes in $S$ is known to at least one player on the team.
    \item $S$ is a subset of a public state. That is, information distinguishing nodes in $S$ is not common knowledge for the team.
    \item For every pair of nodes $h \ne h' \in S$, and every infoset $I \prec h, h'$, the two nodes $h$ and $h'$ must follow from the same action $a$ at $I$, that is, $Ia \prec h, h'$. That is, the information distinguishing nodes in $S$ was not generated by players on the team.
    \item If the team has common external information, then $U^\cabove$ has size $1$ (and its single element is a \pmax-node by assumption), and thus $S$ also can also have size at most $1$.
\end{enumerate}
Conditions~1 and~3 are effectively the negation of the definition of common external information, with the role of infosets $I$ in that definition taken by public states $U$. Thus, in some sense, the reachable width measures the {\em amount of uncommon external information} in the game.

Therefore, \Cref{co:low-external-info-loss} interpolates nicely between the extremes of common external information (which, by Property 4, has reachable width $1$), and the game used in the {\sf NP}-hardness reduction (\Cref{th:hardness}), which can have reachable width $\Theta(N)$. 

Using reachable sets, as opposed to merely arguing about the treewidth $W$ and bounding $\abs{X_U} \le 2^{W+1}$, is crucial in this argument: while Items~1, 2, and~4 in the above discussion follow for unreachable sets as well, Item~3 is false for unreachable sets. Thus, the treewidth $W$ cannot be interpreted as the amount of uncommon {\em external} information, while the reachable width $w$ can. In \Cref{se:counterexample}, we show an example family of games in which the tree decomposition has $O(1)$ reachable width (and thus low uncommon external information), but $\Theta(N)$ treewidth. 

\subsection{Public Actions}
Suppose that our game has the following structure for \pmax: \pmax has $n$ players. Nature begins by picking {\em private types} $t_i \in [t]$ for each player $i \in \pmax$, and informs each player $i$ of $t_i$. From that point forward, all actions are public, the player to act is also public, and no further information is shared between teams. We call such games {\em public action}. For example, poker has this structure.

Assume, again without loss of generality, that the branching factor of the game is at most $2$---this assumption can be satisfied by splitting decision nodes as necessary, and increases the number of public states by only a constant factor.

Consider a public state $U^\cabove \in \mc J^{\cabove}_\pmax$ at which a player $i \in \pmax$ picks one of two actions $L$ or $R$. Since all actions are public, the set of reachable subsets of $U$ can be described as follows: for each type $t_i \in [t]$, $i$ chooses to either play $L$ in $U^\cabove$, play $R$ in $U^\cabove$, or not play to reach $U^\cabove$ at all. For each other player $i' \ne i \in \pmax$, $i'$ choosees, for each type $t_i \in [t]$, whether or not to play to reach $U^\cabove$. There are a total of $3^t 2^{(n-1)t}$ choices that can be made in this fashion, so we have $\abs{X_U} \le 3^t 2^{t(n-1)}$. Thus, we have shown:
\begin{corollary}\label{co:public-actions}
Team correlated equilibria in extensive-form team public-action games with at most $t$ types per player, and $n$ players on each team can be computed via a linear program of size 
$
O(3^t 2^{t(n-1)}NW) \le 2^{O(tn)} N
$.
\end{corollary}
This bound is much tighter than the bound given by the previous section---we have $W, w = O(t^n)$, so \Cref{co:low-external-info-loss} is subsumed by the trivial bound $\abs{X_U} \le 2^{W+1} = 2^{O(t^n)}$. It is also again in some sense tight: the game used in \Cref{th:hardness} has public actions and $t = \Theta(N)$ types for both players.

\section{Experiments}
\begin{table*}[!tb]
\input{experiments.tex}
\caption{Experiments. ``oom'' indicates that our algorithm exhausted the memory limit of 16GB. ``---'' means that the respective paper did not report a runtime for that game. Our runtimes list only the time taken by the LP solver; the time to construct the LP itself is smaller in all instances. LPs are solved with the barrier algorithm in Gurobi 9.1 with crossover and presolver both off (the latter was needed for numerical stability), using 4 CPU cores. The hardware is comparable to that used by FCGS-21, and weaker than that used by ZAC-20. *: The Leduc games in ZAC-20 were constructed using a different implementation than ours and FCGS-21, and thus the reported game sizes differ despite the underlying games being the same. Our implementation matches FCGS-21.
}
\label{ta:experiments}
\end{table*}

We conducted experiments to compare our algorithm to the prior state of the art, namely the algorithms of \citet{Farina21:Connecting} (``FCGS-21'' in the table) and \citet{Zhang20:Computing} (``ZAC-20'' in the table). Experimental results can be found in \Cref{ta:experiments}. The experiments table has the following syntax for identifying games: $mn$G$p$, where $m$ and $n$ are the sizes of teams \pmax and \pmin respectively, G is a letter representing the game, and $p$ represents parameters of the game, described below. All games described are variants of common multi-player games in which teams are formed by colluding players, who we assume will evenly split any reward. For example, ``31K5'' is Kuhn poker (K) with $\abs{\pmax} = 3, \abs{\pmin} = 1$, and $5$ ranks. Where relevant, \pmax consists of the first $m$ players, and \pmin consists of the remaining $n$ players. %
\begin{itemize}
    \item $mn$K$r$ is {\em Kuhn poker} with $r$ ranks. %
    \item $mn$L$brc$ is {\em Leduc poker}. $b$ is the maximum number of bets allowed in each betting round. $r$ is the number of ranks. $c$ is the number of suits (suits are indistinguishable). In variant L', team \pmax is not allowed to raise.
    \item $mn$D$n$ is {\em Liar's Dice} with one $n$-sided die per player. 
    In Liar's Dice, if both teams have consecutive players, then the game value  is triviallyj $0$. Therefore, in these instances, instead of \pmin being the last $n$ players, the last $2n$ players alternate teams---for example, in 42D, the teams are $\pmax = \{ 1, 2, 3, 5\}$ and $\pmin = \{4, 6\}$.
    \item $mn$G and $mn$GL are {\em Goofspiel} with 3 ranks. GL is the limited-information variant.
\end{itemize}
These are the same games used by \citet{Farina21:Connecting} in their work; we refer the interested reader to that paper for detailed descriptions of all the games. In many cases, teams either have size $1$ or have common external information (the latter is always true in Goofspiel). In these cases, it would suffice, after inflation, to use the standard sequence-form representation of the player's strategy space~\cite{Koller96:Efficient}.  However, we run our technique anyway, to demonstrate how it works in such settings.

Our experiments show clear state-of-the-art performance in all tested cases in which comparisons could be made, except Kuhn Poker. In Kuhn Poker, the number of types $t$ is relatively large compared to the game size, so our technique scales poorly compared to prior techniques.

\section{Conclusions}

In this paper, we devised a completely new algorithm for solving team games that uses tree decomposition of the constraints representing each team's strategy to reduce the number and degree of constraints required for correctness.  
Our approach has the bags of the tree decomposition  correspond  to team-public  states---that  is,  minimal sets of  nodes  (that is, states of the team) such  that,  upon reaching the set,  it  is  common  knowledge  among  the  players  on  the  team  that  the set has  been reached.  
Our algorithm reduces the problem of solving team games to a linear program of size $NW^{w+1}$, where $w$ is a parameter that depends on the amount of uncommon external information and $W$ is the treewidth. In {\em public-action games}, we achieve a tighter bound $2^{O(nt)} N$ for teams of $n$ players with $t$ types each.
Our algorithm is based on a new way to write a custom, concise tree decomposition, and its fast run time does not rely on low treewidth.
We show via experiments on a standard suite of games that our algorithm achieves state-of-the-art performance on all benchmark games except Kuhn poker. We also present, to our knowledge, the first experiments for this setting where more than one team has more than one member.

Compared to the techniques of past papers~\cite{Celli18:Computational,Zhang20:Computing,Farina21:Connecting}, our technique has certain clear advantages and disadvantages.

\begin{enumerate}
    \item {\em Advantage}: Unlike prior techniques, ours does not require solving integer programs for best responses. While integer programming and normal-form double oracle can have reasonable practical performance, neither has worst-case performance bounds. In contrast, we are able to derive worst-case performance bounds.
    \item {\em Advantage}: Since our algorithm describes the polytope of correlated strategies directly, we get equilibrium finding in correlated strategies for free---instead of, say, having to run a double oracle algorithm, which, like integer programming, has no known polynomial convergence bound despite reasonable practical performance in some cases.
    \item {\em Advantage}: In domains where there is not much uncommon external information (i.e., $w(U)$ or $t$ is small), our program size scales basically linearly in the game size. Thus, our algorithm is able to tackle some games with $10^5$ sequences for both players.
    \item {\em Disadvantage}: Our algorithm scales poorly in games with high uncommon external information. In the experiments, this can  be seen in the Kuhn poker instances.
    \end{enumerate}

\section*{Acknowledgements}
This material is based on work supported by the National Science Foundation under grants IIS-1718457, IIS-1901403, and CCF-1733556, and the ARO under award W911NF2010081. We also thank Gabriele Farina, Andrea Celli, and our anonymous reviewers at AAAI and AAAI-RLG for suggesting improvements to the writing.

\bibliography{dairefs}

\newpage\onecolumn
\appendix
\section{Proofs}
\subsection{\Cref{th:hardness}}
\begin{proof}
We reduce from 3-SAT. Let $\phi$ be a 3-SAT instance, and consider the following game for a single team $\pmax = \{\text{P1}, \text{P2}\}$. First, nature chooses a clause uniformly at random. Then, P1 is told the clause and chooses a variable within that clause. Finally, P2 is told the variable that P1 chose, but {\em not} the clause, and chooses an assignment (true or false) to that literal. The team wins if P2's assignment satisfies the clause. 

Clearly, the game has size linear in the size of the instance (six terminal nodes per clause). If $\phi$ is satisfiable, the team can always win by having P2 play a satisfying assignment and P1 pick a positive literal in each clause. If $\phi$ is not satisfiable, then no matter what P2's strategy is, the team must lose with probability at least $1/m$ where $m$ is the number of clauses. Therefore, this game has value at least $1-1/2m$ if and only if $\phi$ is satisfiable. 
\end{proof}
\subsection{\Cref{le:jtt-csp}}
\begin{proof}
We repeat the constraint system here for clarity.
\begin{align}
    \vec x_U &= \sum_{\tilde{\vec x} \in X_U} \lambda^U_{\tilde{\vec x}} \tilde{\vec x} &&\qq{$\forall$ bags $U \in \mc J$}\label{eq:apx2}\\
    \sum_{\substack{\tilde{\vec x} \in X_U : \\ \tilde{\vec x}_{U \cap V} = \tilde{\vec x}^*}} \lambda^U_{\tilde{\vec x}} &= \sum_{\substack{\tilde{\vec x} \in X_V : \\ \tilde{\vec x}_{U \cap V} = \tilde{\vec x}^*}} \lambda^V_{\tilde{\vec x}} &&\qq{$\forall$ edges $(U, V)$ in $\mc J$, locally feasible   $\tilde{\vec x}^* \in X_{U \cap V}$}\label{eq:apx3}
\end{align}
Since the $\vec \lambda^U$s are local distributions with consistent marginals across the edges of the tree decomposition, the junction tree theorem gives a distribution $\vec \Lambda$ on $\zo^n$ such that, for every bag $U$, the marginal of $\vec \Lambda$ on $U$ is exactly $\vec \lambda^U$. By definition, we have $\vec x = \E_{\hat{\vec x} \sim \vec \Lambda} \hat{\vec x}$. Consider some $\hat{\vec x} \in \{0,1\}^n$ in the support of $\vec \Lambda$. Then $\hat{\vec x}_U$ is in the support of $\vec \lambda^U$, which in particular means that $\hat{\vec x}$ satisfies all constraints whose variables are fully contained in $U$. Since this is true of every bag $U$, and every constraint has its variables fully contained in some bag $U$, we have $\hat{\vec x} \in X$, and thus $\vec x \in \mc X$, as needed.

We now bound the size (number of total terms) of the system. 
\begin{enumerate}
    \item The number of appearances of each $x_i$ is at most the total sum of all bag sizes, which is $O(W\abs{\mc J})$. 
    \item Each variable $\lambda^U_ {\tilde{\vec x}}$ appears at most $W$ times in constraints of the form \eqref{eq:apx2}, one for each $i$ such that $\tilde{x}_i = 1$. This accounts for $O(W \sum_{U \in \mc J} \abs{X_U})$ terms.
    \item In each edge $(U, V)$, each $\lambda^U_{\tilde{\vec x}}$ and $\lambda^V_{\tilde{\vec x}}$ appears exactly once in constraints of the form~\eqref{eq:apx3}. Thus, each $\lambda^U_{\tilde{\vec x}}$ appears an additional $D$ times, where $D$ is the maximum degree of any bag $U$ of $\mc J$. This accounts for $O(D \sum_{U \in \mc J} \abs{X_U})$ terms.
\end{enumerate}
Adding these up gives the stated result.
\end{proof}
\newpage
\section{Example Game Where Reachable Width is Much Smaller than Width}\label{se:counterexample}

\begin{figure*}[htb]
    \centering
    \input{counterexample.tex}
    \caption{The game tree of the family of games in \Cref{se:counterexample}, for $k=3$. Information sets are connected by dotted lines. Team nodes are black, and nature and terminal nodes are white.}
    \label{fi:counterexample}
\end{figure*}
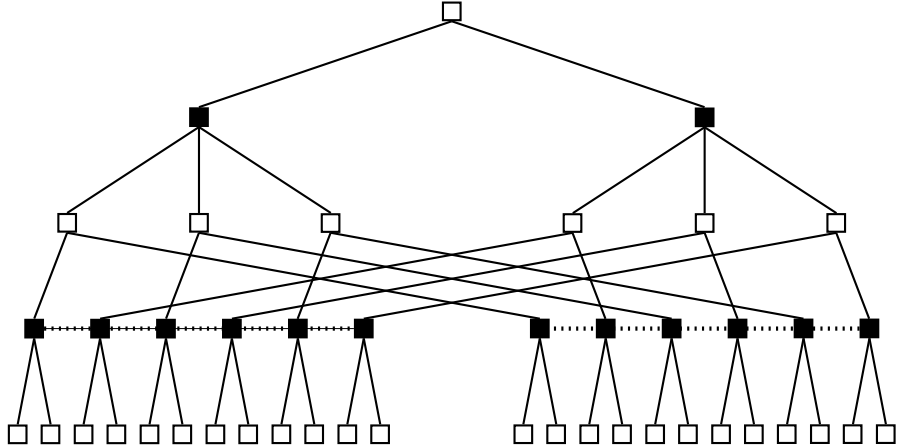
Consider the following family of games involving a team of two players $\pmax = \{\text{P1}, \text{P2}\}$. Nature picks a random type $t_1$ for P1. P1 learns $t_1$ and picks an action $a \in [k]$. Nature then picks a type $t_2 \in \zo$ for P2, possibly conditional on P1's type and/or action. P2 learns $t_2$, but {\em not} the action $a$ played by P1, and must then play one of two actions. We do not specify the rewards as they are not relevant. The game tree for $k=3$ is depicted in \Cref{fi:counterexample}.

In this game, the two P2 infosets, each of size $2k$, are public states. Therefore, the tree decomposition defined in this paper has treewidth $6k-1$ (when $k=3$, the treewidth is $17$). However, there are only two bits of external information, namely the two types $t_1, t_2 \in \zo$. Indeed, using reachable width instead of treewidth fixes this issue: the reachable width of each public state in this game is only $2$ (regardless of the value of $k$), which accurately reflects the amount of uncommon external information. 

We could have written this counterexample without nature picking the type $t_2$, but then the game tree would be rewriteable by switching the order of P1 and P2's decisions, and having P2 decide her move at the root of the tree, upon which the game would become perfect information. Adding the second layer of nature nodes prevents such a straightforward rewriting.

\newpage
\section{Full Constraint System and Example}\label{se:full-program}
In this section, we state the constraint system for the team's decision space induced by chaining \Cref{al:local-feasible} with \eqref{eq:generic-tree-decomposition-constraints}, and give an example using the game in \Cref{fi:example}.
\subsection{General Constraint System}
For each public state $U^\cabove \in \mc J_\pmax^\cabove$, and for each reachable subset $S \subseteq U^\cabove$, let $I_S$ be the set of \pmax-infosets intersecting $S$. For each joint action $\vec a \in \bigtimes_{I \in \mc I_S}A(I)$, let $S \vec a \subseteq U^\cbelow$ be the set of children of $S$ reached if the team plays joint action $\vec a$. That is,
$$S \vec a = \{ ha_I : h \in S, h \in I \in {\mc I}_S\} \cup \{ ha : h \in S, P(h) \ne \pmax, a \in A(h)\}.$$
Let $\tilde{\vec x}^{S\vec a} \in \zo^U$ be the indicator function on $S \cup S\vec a$, i.e.,
$\tilde{\vec x}^{S\vec a}_i = \vec 1\{i \in S \cup S\vec a\}$.  For each such set $S \cup S \vec a \subseteq U$, we introduce a variable $\lambda_{S \vec a}$ representing the local probability mass on $\tilde{\vec x}^{S \vec a}$.

In the root bag\footnote{The condition that $\vec \lambda$ is a local distribution in all other bags will follow from the marginalization constraints between bags.} $\{\textsc{Root}\}$, we introduce the constraint that $\vec \lambda$ is a feasible local distribution in this bag and $x_\textsc{Root}$ is consistent with it, which simply means that $x_\textsc{Root} = \lambda_{\{\textsc{Root}\}} = 1$.

In every bag, we introduce the constraint that $\vec x$ is consistent with the choice of the local distributions\footnote{The constraints for $h \in U^\cabove$ again follow from marginalization constraints.}:
\begin{align}
    x_h = \sum_{\substack{ \text{reachable }S \vec a \subseteq U \\ h \in S \vec a}} \lambda_{S \vec a} \qq{$\forall\ h \in U^\cbelow$.}
\end{align}
Finally, if $U$ is not the root bag, for each nonempty\footnote{The constraints for empty $S$ are redundant.} reachable subset $S \subseteq U^\cabove$, we introduce the constraint
\begin{align}
    \sum_{\substack{ \text{reachable }S' \vec a'  \subseteq U' \\ S = S' \vec a' \cap U }} \lambda_{S' \vec a'} &= \sum_{\vec a \in \bigtimes_{I \in \mc I_S}A(I)} \lambda_{S \vec a}
\end{align}
where $U'$ is the parent of $U$. This constraint states that, for each reachable subset $S \subseteq U^\cabove$, the ``marginal'' onto $S$ in both directions is consistent: from above, this marginal is constructed by summing the probabilities of all $S' \vec a'$ in the parent node such that the marginal of $S' \vec a'$ onto $U$ is exactly $S$; from below, it is constructed by summing over all joint actions $\vec a$ available to the team at $S$. As stated in the main paper, in the case where the team has perfect recall, the left-hand side of this constraint will have exactly one term, and we recover, up to some redundancies, the sequence-form polytope of the team.

\subsection{Example: Constraint System for the Game in \Cref{fi:example}}
\newcommand{\longline}{\hphantom{\quad\lambda_\text{BCDF} + \lambda_\text{BCDG}}}
\begin{align}
\minibox[c,frame]{A} \text{ and } \minibox[c,frame]{A \\ BC}&~~~~\begin{aligned}
     \longline \llap{$\lambda_\text{A}$} &= \lambda_\text{ABC} = x_\text{A} = x_\text{B} = x_\text{C} = 1 \\
\end{aligned}\\
\minibox[c,frame]{BC \\ DEFG}&\left\{\begin{aligned}
    \longline \llap{$\lambda_\text{ABC}$}
&= \lambda_\text{BCDF} + \lambda_\text{BCDG} + \lambda_\text{BCEF} + \lambda_\text{BCEG} \\
    x_\text{D} &= \lambda_\text{BCDF} + \lambda_\text{BCDG} \\
    x_\text{E} &= \lambda_\text{BCEF} + \lambda_\text{BCEG} \\
    x_\text{F} &= \lambda_\text{BCDF} + \lambda_\text{BCEF} \\
    x_\text{G} &= \lambda_\text{BCDG} + \lambda_\text{BCEG}
\end{aligned}\right.\\
\minibox[c,frame]{D \\ HI}&\left\{\begin{aligned}
   \quad \lambda_\text{BCDF} + \lambda_\text{BCDG} &= \lambda_\text{DH} + \lambda_\text{DI}\\
    x_\text{H} &= \lambda_\text{DH}\\
    x_\text{I} &= \lambda_\text{DI}
\end{aligned}\right.\\
\minibox[c,frame]{EF \\ JKLM}&\left\{\begin{aligned}
\longline \llap{
    $\lambda_\text{BCDF}$} &= \lambda_\text{FL} + \lambda_\text{FM}\\
    \lambda_\text{BCEF} &= \lambda_\text{EFJL} + \lambda_\text{EFKM}\\
    \lambda_\text{BCEG} &= \lambda_\text{EJ} + \lambda_\text{EK}\\
    x_\text{J} &= \lambda_\text{EFJL} + \lambda_\text{EJ}\\
    x_\text{K} &= \lambda_\text{EFKM} + \lambda_\text{EK}\\
    x_\text{L} &= \lambda_\text{EFJL} + \lambda_\text{FL}\\
    x_\text{M} &= \lambda_\text{EFKM} + \lambda_\text{FM}\\
\end{aligned}\right.
\end{align}

\end{document}

%% file: picture.tex
\begin{tikzpicture}[
    level 1+/.style={level distance=20pt},
   every tree node/.style = {shape=rectangle, draw, fill=black!25},
    every path/.style={-, thick},
    nature/.style={fill=none},
    baseline={(current bounding box.center)},
    ]
    \small
    \Tree[
         .\node[nature](1){A}; 
        \edge []; [.\node[](2){B};
            \edge []; [.\node[](4){D};
                \edge []; \node[nature](8){H};
                \edge []; \node[nature](9){I};
            ]
            \edge []; [.\node[](5){E};
                \edge []; \node[nature](10){J};
                \edge []; \node[nature](11){K};
            ]
        ]
        \edge []; [.\node[](3){C};
            \edge []; [.\node[](6){F};
                \edge []; \node[nature](12){L};
                \edge []; \node[nature](13){M};
            ]
            \edge []; \node[nature](7){G};
        ]
    ]
    \draw[dotted, ultra thick, -] (5) -- (6);
    \end{tikzpicture}
    \qq{$\longrightarrow$}
    \begin{tikzpicture}[
    level 1/.style={level distance=22pt},
    level 2+/.style={level distance=30pt},
    every tree node/.style = {shape=rectangle, draw},
    every path/.style={-, thick},
    nature/.style={fill=none},
    baseline={(current bounding box.center)}
    ]
    \small
    \Tree[
    .\node[align=center](x){A}; \edge []; [
         .\node[align=center](x){A \\ BC}; 
        \edge []; [.\node[align=center](x){BC \\ DEFG};
            \edge []; \node[align=center](x){D \\ HI};
            \edge []; [.\node[align=center](x){EF \\ JKLM}; \edge[draw=none]; \node[draw=none](x){};]
        ]
    ]
    ]
    \end{tikzpicture}

%% file: experiments.tex
\definecolor{light-gray}{gray}{0.65}
\newcommand{\G}[1]{\color{light-gray}$#1$}

\newcommand{\tablesize}{\footnotesize}\setlength{\tabcolsep}{5pt}
\newcommand{\mymidrule}{\cmidrule(lr){1-2} \cmidrule(lr){3-5} \cmidrule(lr){6-8}  \cmidrule(lr){9-10} \cmidrule(lr){11-13}}
\centering{
\tablesize
\begin{tabular}{lr rr rr rr rr rrr}
	\toprule
	 && \multicolumn{3}{c}{Team \pmax} & \multicolumn{3}{c}{Team \pmin} &&& \multicolumn{3}{c}{Runtime}\\
		{Game} & {$\abs{Z}$} & {\#Seq} & {$\sum \abs{X_C}$} & {Ratio} & {\#Seq} & {$\sum \abs{X_C}$} & {Ratio} & {Value} & {LP Size} & {Ours} & {FCGS-21} & {ZAC-20} \\
	
	\mymidrule
	
	21K3  & 78 & 91 & 351 & 3.9 & 25 & 25 & 1 & .0000 & 2386 & {\bf 0.01s} & {\bf 0.01s} & \G<0.7s \\
	
    21K4  & 312 & 177 & 1749 & 9.9 & 33 & 33 & 1 & -.0417 & 18810 & 0.02s & {\bf 0.01s} & 0.7s \\
    
    21K6  & 1560 & 433 & 52669 & 122 & 49 & 49 & 1 & -.0236 & 1150838 & 1.37s & --- & 2s\\
    
    21K8  & 4368 & 801 & 1777061 & 2218 & 65 & 65 & 1 & -.0193 & 62574750 & 2m24s & \G<4.96s  & 4s\\
    
    21K12  & 17160 & 1873 &  && 97 & 97 & 1 &  & & oom & {\bf 4.96s} & 10s \\
    
    22K5  & 3960 & 611 & 23711 & 39 & 3083 & 521 & 5.9 & -.0368 & 426297 & 1.27s & --- & --- \\
    
    31K5  & 3960 & 2611 & 974470 & 373 & 81 & 81 & 1 & -.0300 & 21106658 & 3m20s & --- & ---\\

    \mymidrule

    21L133* & 6477 & 2725 & 17718 & 6.5 & 457 & 703 & 1.5 & .2148 & 126075 & {\bf 0.49s} & 2m23s & 1h22m \\
    
    21L143 & 20856 & 6377 & 115281 & 18 & 801 & 1225 & 1.5 & .1073 & 1195766 & {\bf 7.66s} & 1h07m & \G>1h22m \\
    
    21L151* & 10020 & 6051 & 130359 & 22 & 1001 & 1531 & 1.5 & -.0192 & 1425583 & 8.02s & --- & 4h24m \\
    
    21L153 & 51215 & 12361 & 757884 & 61 & 1241 & 1891 & 1.5 & .0240 & 11234573 & {\bf 3m50s} & \G>1h07m & \G>4h24m \\
     
    21L223 & 8762 & 5765 & 21729 & 3.8 & 1443 & 3123 & 2.2 & .5155 & 112305 & {\bf 0.23s} & 3m29s & --- \\
    
    21L523 & 775148 & 492605 & 2042641 & 4.1 & 123153 & 305835 & 2.5 & .9520 & 10507398 & {\bf 3m27s} & \G>3m29s & --- \\
    
    22L133 & 80322 & 9781 & 55788 & 5.7 & 9745 & 52053 & 5.3
 & .1470 & 785032 & {\bf 9.26s} & \G>2m23s & \G>1h22m \\
     
    31L'132* & 3834 & 3991 & 46122 & 12 & 151 & 235 & 1.6 & -.1333 & 369673 & 2.97s & --- & 1m08s \\
    
    31L'133* & 8898 & 5644 & 69642 & 12 & 151 & 235 & 1.6 & -.1457 & 577942 & 8.88s & --- & 21m04s \\
    
    31L133 & 80322 & 42361 & 703390 & 17 & 1633 & 2479 & 1.5 & .3470 & 6326658 & {\bf 1m52s} & \G>2m23s & \G>1h22m \\

    \mymidrule
    
    21D3 & 13797 & 6085 & 74635 & 12 & 1021 & 2686 & 2.6 & .2840 & 500665 & {\bf 2.03s} & 1m12s & --- \\
    
    21D4 & 262080 & 87217 & 3521398 & 40 & 10921 & 29749 & 2.7 & .2843 & 32755273 & {\bf 9m59s} & \G>6h & --- \\
    
    22D3 & 331695 & 36784 & 402669 & 11 & 36388 & 381218 & 11 & .2000 & 5275196 & {\bf 1m00s} & \G>6h & --- \\
    
    33D2 & 262080 & 31545 & 178123 & 5.6 & 29433 & 152286 & 5.2 & .0721 & 1862387 & 15.7s & --- & --- \\
    
    42D2 & 262080 & 83969 & 761600 & 9.1 & 9491 & 30150 & 3.2 & .2647 & 4751391 & 1m24s & --- & --- \\
    
    51D2 & 262080 & 185905 & 2537927 & 14 & 3459 & 7207 & 2.1 & .3333 & 17635669 & 8m21s & --- & --- \\

    \mymidrule
    
    21GL & 1296 & 2713 & 8572 & 3.2 & 934 & 2158 & 2.3 & .2524 & 34310 & {\bf 0.06s} & 0.33s & --- \\
    
    21G & 1296 & 3601 & 11344 & 3.2 & 1630 & 3766 & 2.3 & .2534 & 47810 & {\bf 0.08s} & 1.18s & --- \\
    
    31GL & 7776 & 21082 & 71278 & 3.4 & 3502 & 7582 & 2.2 & .2803 & 249374 & 0.66s & --- & --- \\
    
    31G & 7776 & 30250 & 102118 & 3.4 & 8758 & 18994 & 2.2 & .2803 & 378410 & {\bf 0.79s} & \G>1.18s & ---\\
    
    32GL & 46656 & 71722 & 241030 & 3.4 & 34393 & 104908 & 3.1 & .0000 & 1095542 & 3.19s & --- & --- \\
    
    32G & 46656 & 160498 & 538546 & 3.4 & 102097 & 310720 & 3.0 & .0000 & 2620502 & 7.86s & --- & ---\\

    \bottomrule
	\end{tabular}
}

%% file: counterexample.tex
\begin{tikzpicture}[
    level 1+/.style={level distance=40pt, sibling distance=5pt},
    every level 0 node/.style={draw, shape=rectangle, fill=black},
    every level 1 node/.style={draw, shape=rectangle, fill=black},
    every level 2 node/.style={draw, shape=rectangle, fill=none},
    every level 3 node/.style={draw, shape=rectangle, fill=black},
    every level 4 node/.style={draw, shape=rectangle, fill=none},
    every path/.style={-, thick},
    nature/.style={fill=none},
    off/.style={draw=none, fill=none},
    baseline={(current bounding box.center)},
    ]
    
    \Tree[
        .\node[nature]{};
            [ .{} 
                [ .\node[nature](l1){}; 
                    [ .\node[](li){}; {} {} ] \edge[draw=none]; [ .\node[](l1b){}; {} {} ]
                ]
                [ .\node[nature](l2){}; 
                    [ .{} {} {} ] \edge[draw=none]; [ .\node[](l2b){}; {} {} ]
                ]
                [ .\node[nature](l3){}; 
                    [ .{} {} {} ] \edge[draw=none]; [ .\node[](l3b){}; {} {} ]
                ]
            ] 
            \edge[draw=none]; [ .\node[off]{}; \edge[draw=none]; [ .\node[off]{}; \edge[draw=none]; \node[off]{\quad\quad\quad\quad}; ] ]
            [ .{} 
                [ .\node[nature](r1){}; 
                    \edge[draw=none]; [ .\node[](r1b){}; {} {} ] [ .{} {} {} ] 
                ]
                 [ .\node[nature](r2){}; 
                    \edge[draw=none]; [ .\node[](r2b){}; {} {} ] [ .{} {} {} ] 
                ]
                 [ .\node[nature](r3){}; 
                    \edge[draw=none]; [ .\node[](r3b){}; {} {} ] [ .\node[](ri){}; {} {} ] 
                ]
            ]
    ]
    \draw[] (l1.south) -- (r1b.north);
    \draw[] (l2.south) -- (r2b.north);
    \draw[] (l3.south) -- (r3b.north);
    
    \draw[] (r1.south) -- (l1b.north);
    \draw[] (r2.south) -- (l2b.north);
    \draw[] (r3.south) -- (l3b.north);
    
    \draw[dotted, ultra thick, -, fill=black, draw=black] (li) -- (l3b);
    \draw[dotted, ultra thick, -] (ri) -- (r1b);
    \end{tikzpicture}

%% file: main.bbl
\begin{thebibliography}{19}
\providecommand{\natexlab}[1]{#1}

\bibitem[{Basilico et~al.(2017)Basilico, Celli, De~Nittis, and
  Gatti}]{Basilico17:Team}
Basilico, N.; Celli, A.; De~Nittis, G.; and Gatti, N. 2017.
\newblock Team-maxmin equilibrium: efficiency bounds and algorithms.
\newblock In \emph{AAAI Conference on Artificial Intelligence (AAAI)}.

\bibitem[{Bowling et~al.(2015)Bowling, Burch, Johanson, and
  Tammelin}]{Bowling15:Heads}
Bowling, M.; Burch, N.; Johanson, M.; and Tammelin, O. 2015.
\newblock Heads-up Limit Hold'em Poker is Solved.
\newblock \emph{Science}, 347(6218).

\bibitem[{Brown and Sandholm(2018)}]{Brown18:Superhuman}
Brown, N.; and Sandholm, T. 2018.
\newblock Superhuman {AI} for heads-up no-limit poker: {Libratus} beats top
  professionals.
\newblock \emph{Science}, 359(6374): 418--424.

\bibitem[{Brown and Sandholm(2019)}]{Brown19:Superhuman}
Brown, N.; and Sandholm, T. 2019.
\newblock Superhuman {AI} for multiplayer poker.
\newblock \emph{Science}, 365(6456): 885--890.

\bibitem[{Celli and Gatti(2018)}]{Celli18:Computational}
Celli, A.; and Gatti, N. 2018.
\newblock Computational results for extensive-form adversarial team games.
\newblock In \emph{AAAI Conference on Artificial Intelligence (AAAI)}.

\bibitem[{Chu and Halpern(2001)}]{Chu01:NP}
Chu, F.; and Halpern, J. 2001.
\newblock On the {NP}-completeness of finding an optimal strategy in games with
  common payoffs.
\newblock \emph{International Journal of Game Theory}.

\bibitem[{Daskalakis and Papadimitriou(2006)}]{Daskalakis06:Computing}
Daskalakis, C.; and Papadimitriou, C. 2006.
\newblock Computing pure {N}ash equilibria in graphical games via {M}arkov
  random fields.
\newblock In \emph{ACM Conference on Electronic Commerce (ACM-EC)}, 91--99. Ann
  Arbor, MI.

\bibitem[{Farina et~al.(2018)Farina, Celli, Gatti, and
  Sandholm}]{Farina18:ExAnte}
Farina, G.; Celli, A.; Gatti, N.; and Sandholm, T. 2018.
\newblock Ex Ante Coordination and Collusion in Zero-Sum Multi-Player
  Extensive-Form Games.
\newblock In \emph{Conference on Neural Information Processing Systems
  (NeurIPS)}.

\bibitem[{Farina et~al.(2021)Farina, Celli, Gatti, and
  Sandholm}]{Farina21:Connecting}
Farina, G.; Celli, A.; Gatti, N.; and Sandholm, T. 2021.
\newblock Connecting Optimal Ex-Ante Collusion in Teams to Extensive-Form
  Correlation: Faster Algorithms and Positive Complexity Results.
\newblock In \emph{International Conference on Machine Learning}.

\bibitem[{Farina and Sandholm(2020)}]{Farina20:Polynomial}
Farina, G.; and Sandholm, T. 2020.
\newblock Polynomial-Time Computation of Optimal Correlated Equilibria in
  Two-Player Extensive-Form Games with Public Chance Moves and Beyond.
\newblock In \emph{Conference on Neural Information Processing Systems
  (NeurIPS)}.

\bibitem[{Kaneko and Kline(1995)}]{Kaneko95:Behavior}
Kaneko, M.; and Kline, J.~J. 1995.
\newblock Behavior strategies, mixed strategies and perfect recall.
\newblock \emph{International Journal of Game Theory}, 24(2): 127--145.

\bibitem[{Koller, Megiddo, and {von Stengel}(1994)}]{Koller94:Fast}
Koller, D.; Megiddo, N.; and {von Stengel}, B. 1994.
\newblock Fast algorithms for finding randomized strategies in game trees.
\newblock In \emph{{ACM} {S}ymposium on {T}heory of {C}omputing ({STOC})}.

\bibitem[{Koller, Megiddo, and {von Stengel}(1996)}]{Koller96:Efficient}
Koller, D.; Megiddo, N.; and {von Stengel}, B. 1996.
\newblock Efficient Computation of Equilibria for Extensive Two-Person Games.
\newblock \emph{Games and Economic Behavior}, 14(2).

\bibitem[{Kova{\v{r}}{\'\i}k et~al.(2021)Kova{\v{r}}{\'\i}k, Schmid, Burch,
  Bowling, and Lis{\`y}}]{Kovavrik21:Rethinking}
Kova{\v{r}}{\'\i}k, V.; Schmid, M.; Burch, N.; Bowling, M.; and Lis{\`y}, V.
  2021.
\newblock Rethinking formal models of partially observable multiagent decision
  making.
\newblock \emph{Artificial Intelligence}, 103645.

\bibitem[{Luce and Raiffa(1957)}]{Luce57:Games}
Luce, R.~D.; and Raiffa, H. 1957.
\newblock \emph{Games and Decisions}.
\newblock New York: John Wiley and Sons.
\newblock Dover republication 1989.

\bibitem[{McMahan, Gordon, and Blum(2003)}]{McMahan03:Planning}
McMahan, H.~B.; Gordon, G.~J.; and Blum, A. 2003.
\newblock Planning in the presence of cost functions controlled by an
  adversary.
\newblock In \emph{International Conference on Machine Learning (ICML)},
  536--543.

\bibitem[{Robertson and Seymour(1986)}]{Robertson86:Graph}
Robertson, N.; and Seymour, P.~D. 1986.
\newblock Graph minors. {II}. {A}lgorithmic aspects of tree-width.
\newblock \emph{Journal of Algorithms}, 7(3): 309--322.

\bibitem[{Wainwright and Jordan(2008)}]{Wainwright08:Graphical}
Wainwright, M.~J.; and Jordan, M.~I. 2008.
\newblock \emph{Graphical models, exponential families, and variational
  inference}.
\newblock Now Publishers Inc.

\bibitem[{Zhang, An, and {\v{C}}ern{\`y}(2020)}]{Zhang20:Computing}
Zhang, Y.; An, B.; and {\v{C}}ern{\`y}, J. 2020.
\newblock Computing Ex Ante Coordinated Team-Maxmin Equilibria in Zero-Sum
  Multiplayer Extensive-Form Games.
\newblock In \emph{International Conference on Machine Learning (ICML)}.

\end{thebibliography}
